\documentclass[sigconf]{acmart}

\usepackage{booktabs} 
\usepackage[T1]{fontenc}
	\usepackage{amsmath}
	\usepackage{amsfonts}
	\usepackage{amsthm}
	\usepackage{mathtools}
	\usepackage{setspace}
	\usepackage{algorithm}
	\usepackage[noend]{algpseudocode}
	\usepackage{graphicx}
\usepackage{geometry}
\usepackage{layout}
\usepackage[latin1]{inputenc}
\usepackage{lineno}

\usepackage{latexsym, soul, color, type1cm, dsfont, float, fancyhdr}
\usepackage[official]{eurosym}

	\newcommand{\ecc}[1]{\texttt{ecEnc}\left( #1 \right)}
	
	\newcommand{\bjconst}{C}
	
	\newcommand{\amdconst}{C_a}
	\newcommand{\naorconst}{C_h}
	\newcommand{\eccconst}{C_e}

	\newcommand{\ball}{polynomial evaluation tuple}
	
	\newcommand{\balls}{polynomial evaluation tuples}
	
	\newcommand{\bin}{field element}
	
	\newcommand{\bins}{field elements}
	
	\newcommand{\ecci}[1]{\texttt{ecDec}\left( #1 \right)}
	
	\newcommand{\getpoly}[3]{\texttt{GetPolynomial}\left(#1, #2, #3\right)}
	
	\newcommand{\fp}[4]{\texttt{h}\left( #1, #2, #3, #4 \right)}
	
	\newcommand{\listen}[1]{\texttt{Listen}\left( #1 \right)}
	
	\newcommand{\silence}[1]{\texttt{IsSilence}\left( {#1} \right)}

	\newcommand{\zero}[1]{\mathbf{0}_{#1}}
	
	\newcommand{\terminate}[0]{\textbf{Terminate}}
	
	\newcommand{\amdenc}[2]{\texttt{amdEnc}\left( #1, #2 \right)}

	\newcommand{\amddec}[2]{\texttt{amdDec}\left( #1, #2 \right)}

	\newcommand{\iscodeword}[2]{\texttt{IsCodeword}\left( #1, #2 \right)}
	
	\newcommand{\ie}{\textit{i}.\textit{e}.}

	\newtheorem{mythm}{Theorem}
	
	\numberwithin{equation}{section}
	\numberwithin{mydef}{section}
	\numberwithin{mythm}{section}

\setcopyright{rightsretained}

\acmConference[Submitted to ICDCN'18]{$19^{th}$ International Conference on Distributed Computing and Networking}{January 2018}{Varanasi, India} 
\acmYear{2018}
\copyrightyear{2018}

\begin{document}
\title{Sending a Message with Unknown Noise}

\author{Abhinav Aggarwal}
\affiliation{%
  \institution{University of New Mexico}
  \streetaddress{Albuquerque, New Mexico, USA}}
\email{abhiag@unm.edu}

\author{Varsha Dani}
\affiliation{%
  \institution{University of New Mexico}
  \streetaddress{Albuquerque, New Mexico, USA}}
\email{varsha@cs.unm.edu}

\author{Thomas P. Hayes}
\affiliation{%
  \institution{University of New Mexico}
  \streetaddress{Albuquerque, New Mexico, USA}}
\email{hayes@cs.unm.edu}

\author{Jared Saia}
\affiliation{%
  \institution{University of New Mexico}
  \streetaddress{Albuquerque, New Mexico, USA}}
\email{saia@cs.unm.edu}

\begin{abstract}
Alice and Bob are connected via a two-way channel, and Alice wants to send a message of $L$ bits to Bob.  An adversary flips an arbitrary but finite number of bits, $T$, on the channel.  This adversary knows our algorithm and Alice's message, but does not know any private random bits generated by Alice or Bob, nor the bits sent over the channel, except when these bits can be predicted by knowledge of Alice's message or our algorithm.  We want Bob to receive Alice's message and for both players to terminate, with  error probability at most $\delta > 0$, where $\delta$ is a parameter known to both Alice and Bob.  Unfortunately, the value  $T$ is unknown in advance to either Alice or Bob, and the value $L$ is unknown in advance to Bob. 

We describe an algorithm to solve the above problem while sending an expected 
$L + O \left( T + \min \left(T+1,\frac{L}{\log L} \right) \log \left( \frac{L}{\delta} \right) \right)$ bits.  A special case is when $\delta = O(1/L^c)$, for some constant $c$.  Then when $T = o(L/\log L)$, the expected number of bits sent is $L + o(L)$, and when $T = \Omega(L)$, the expected number of bits sent is $L + O\left( T \right)$, which is asymptotically optimal.
\end{abstract}

%
%

\ccsdesc[500]{Mathematics of computing~Information theory}
\ccsdesc[500]{Computing methodologies~Distributed algorithms}
\ccsdesc[300]{Security and privacy~Security protocols}

\keywords{Reed Solomon Codes, Interactive Communication, Adversary, Polynomial, AMD Codes, Error Correction Codes, Fingerprinting}

\maketitle

\section{Introduction}	

What if we want to send a message over a noisy two-way channel, and little is known in advance?  In particular, imagine that Alice wants to send a message to Bob, but the number of bits flipped on the channel is unknown to either Alice or Bob in advance.  Further, the length of Alice's message is also unknown to Bob in advance.  While this scenario seems like it would occur quite frequently, surprisingly little is known about it.

In this paper, we describe an algorithm to efficiently address this problem.  To do so, we make a critical assumption on the type of noise on the channel.  We assume that an adversary flips bits on the channel, but this adversary is not completely omniscient.  The adversary knows our algorithm and Alice's message, but it \emph{does not} know the private random bits of Alice and Bob, nor the bits that are sent over the channel, except when these bits do not depend on the random bits of Alice and Bob.  Some assumption like this is necessary : if the adversary knows all bits sent on the channel and the number of bits it flips is unknown in advance, then no algorithm can succeed with better than constant probability (see Theorem 6.1 from~\cite{icalp15} for details\footnote{Essentially, in this case, the adversary can run a man-in-the-middle attack to fool Bob into accepting the wrong message}).

Our algorithm assumes that a desired error probability, $\delta > 0$ is known to both Alice and Bob, that the adversary flips some number $T$ bits that is finite but unknown in advance, and that the length of Alice's message, $L$ is unknown to Bob in advance. Our main result is then summarized in the following theorem.


	\begin{mythm}
		\label{thm:mainThm}
		Our algorithm tolerates an unknown number of adversarial errors, $T$, and for any $\delta > 0$, succeeds in sending a message of length $L$ with probability at least $1-\delta$, and sends an expected $L + O \left( T + \min \left(T+1,\frac{L}{\log L} \right) \log \left( \frac{L}{\delta} \right) \right)$  bits.	
				\end{mythm}

An interesting case to consider is when the error probability is polynomially small in $L$, i.e. when $\delta = O(1/L^c)$, for some constant $c$.  Then when $T = o(L/\log L)$, our algorithm sends $L + o(L)$ expected bits.  When $T = \Omega(L)$, the number of bits sent is $L + O\left( T \right)$, which is asymptotically optimal.


 \subsection{Related Work}

\noindent
\textbf{Interactive Communication} Our work is related to the area of interactive communication.  The problem of interactive communication asks how two parties can run a protocol $\pi$ over a noisy channel.  This problem was first posed by Schulman~\cite{schulman:deterministic,schulman:communication}, who describes a deterministic method for simulating interactive protocols on noisy channels with only a constant-factor increase in the total communication complexity. This initial work spurred vigorous interest in the area (see~\cite{braverman:coding} for an excellent survey).

\par Schulman's scheme tolerates an adversarial noise rate of $1/240$, even if the adversary is \emph{not} oblivious.  It critically depends on the notion of a {\it tree code} for which an exponential-time construction was originally provided. This exponential construction time motivated work on  more efficient constructions~\cite{braverman:towards-deterministic,peczarski:improvement,moore:tree}.  There were also efforts to create alternative codes~\cite{gelles:efficient,ostrovsky:error}.  Recently, elegant computationally-efficient schemes that tolerate a constant adversarial noise rate have been demonstrated~\cite{brakerski:efficient,ghaffari:optimal2}. Additionally, a large number of results have improved the tolerable adversarial noise rate~\cite{brakerski:fast,braverman:towards,ghaffari:optimal,franklin:optimal,braverman:list}, as well as tuning the communication costs to a known, but not necessarily constant, adversarial noise rate~\cite{haeupler2014interactive}.\\

\noindent
\textbf{Interactive Communication with Private Channels}	 Our paper builds on a recent result on interactive communication by Dani et al~\cite{icalp15}.  The model in~\cite{icalp15} is equivalent to the one in this paper except that 1) they assume that Alice and Bob are running an arbitrary protocol $\pi$; and 2) they assume that both Alice and Bob know the number of bits sent in $\pi$.  In particular, similar to this paper, they assume that the adversary flips an unknown number of bits $T$, and that the adversary does not know the private random bits of Alice and Bob, or the bits sent over the channel.

If the protocol $\pi$ just sends $L$ bits from Alice to Bob, then the algorithm from~\cite{icalp15} can solve the problem we consider here.  In that case, the algorithm of~\cite{icalp15} will send an expected $L + O \left( \sqrt{L(T+1)\log L} + T \right)$ bits, with a probability of error that is $O(1/L^c)$ for any fixed constant $c$.  

For the same probability of error, the algorithm in this paper sends an expected $L + O(\min((T+1)\log L),L) + T)$ bits.  This is never worse than~\cite{icalp15}, and can be significantly better.  For example, when $T=O(1)$, our cost is $L + O(\log L)$ versus $L + O(\sqrt{L\log L})$ from~\cite{icalp15}.  In general if $T = o(L/\log L)$ our cost is asymptotically better than~\cite{icalp15}.  Additionally, unlike~\cite{icalp15}, the algorithm in this paper does not assume that $L$ is known in advance by Bob.

An additional results of~\cite{icalp15} is a theorem showing that private channels are necessary in order to tolerate unknown $T$ with better than constant probability of error.\\

\noindent
\textbf{Rateless Codes}  Rateless error correcting codes enable generation of potentially an infinite number of encoding symbols from a given set of source symbols with the property that given any subset of a sufficient number of encoding symbols, the original source symbols can  be recovered. Fountain codes~\cite{mackay2005fountain,mitzenmacher2004digital} and LT codes~\cite{palanki2004rateless,luby2002lt,hashemi2014near} are two classic examples of rateless codes. 	Erasure codes employ feedback for stopping transmission~\cite{palanki2004rateless,luby2002lt} and for error detection~\cite{hashemi2014near} at the receiver.

Critically, the feedback channel, i.e. the channel from Bob to Alice, is typically assumed to be noise free. We differ from this model in that we allow noise on the feedback channel, and additionally, we tolerate bit flips, while most rateless codes tolerate only bit erasures.

	\subsection{Formal Model}

\paragraph{Initial State}
	We assume that Alice initially knows some message $M$ of length $L$ bits that she wants to communicate to Bob, and that both Alice and Bob know an error tolerance parameter $\delta >0$.  However, Bob does not know $L$ or any other information about $M$ initially.  Alice and Bob are connected by a two-way binary communication channel.

\paragraph{The Adversary} We assume an adversary can flip some \textit{a priori} unknown, but finite number of bits $T$ on the channel from Alice to Bob or from Bob to Alice.  This adversary knows $M$, and all of our algorithms.  However, it does not know any random bits generated by Alice or Bob, or the bits sent over the channel, except when these can be determined from other known information.
			
	\paragraph{Channel steps}
	We assume that communication over the channel is synchronous. A \emph{channel step} is defined as the amount of time that it takes to send one bit over the channel.  As is standard in distributed computing, we assume that all local computation is instantaneous.
		
	\paragraph{Silence on the channel}
	Similar to~\cite{icalp15}, when neither Alice nor Bob sends in a channel step, we say that the channel is silent. In any contiguous sequence of silent channel steps, the bit received on the channel in the first step is set by the adversary for free. By default, the bit received in the subsequent steps of the sequence remains the same, unless the adversary pays for one bit flip each time it wants to change the value of the bit received.
	
	\subsection{Paper organization}
	The rest of the paper is organized as follows. We first discuss an algorithm for the case when both Alice and Bob share the knowledge of $L$ in Section~\ref{sec:knownL}. We present the analysis for failure probability, correctness, termination and number of bits sent by this algorithm in Section~\ref{sec:analysisKnownL}. Then, we remove the assumption of knowledge of $L$ and provide an algorithm for the unknown $L$ case in Section~\ref{sec:unknownL}, along with its analysis. Finally, in Section~\ref{sec:conclusion}, we conclude the paper by stating the main result and discuss some open problems.
	
	\section{Known $L$}
	\label{sec:knownL}
	We first discuss the case when Bob knows $L$.  We remove this assumption later in Section~\ref{sec:unknownL}.

	Our algorithm makes critical use of Reed-Solomon codes from \cite{reedSolomon}. Alice begins by encoding her message using a polynomial of degree $d = \lceil L/\log q \rceil-1$ over $GF(q)$, where $q = 2^{\lceil \log L \rceil}$. She sends the values of this polynomial computed at certain elements of the field as message symbols to Bob. Upon receiving an appropriate number of these points, Bob computes the polynomial using the Berlekamp-Welch algorithm \cite{welch1986error} and sends a fingerprint of his guess to Alice. Upon hearing this fingerprint, if Alice finds no errors, she echoes the fingerprint back to Bob, upon receiving a correct copy of which, Bob terminates the algorithm. Unless the adversary corrupts many bits, Alice terminates soon after.
	
	\par However, in the case where Alice does not receive a correct fingerprint of the polynomial from Bob, she sends two more evaluations of the polynomial to Bob. Bob keeps receiving extra evaluations and recomputing the polynomial until he receives the correct fingerprint echo from Alice.
	
	\subsection{Notation}
	Some helper functions and notation used in our algorithm are described in this section. We denote by $s \in_{\text{u.a.r.}}S$ the fact that $s$ is sampled uniformly at random from the set $S$.
	\paragraph{Fingerprinting} For fingerprinting, we use a well known theorem by Naor and Naor \cite{naor1993small}, slightly reworded as follows:
	
	\begin{mythm}
	\label{thm:hash}
		\cite{naor1993small} Fix integer $\ell > 0$ and real $p \in (0,1)$. Then there exist constants $C_s, \naorconst > 0$ and algorithm \texttt{h} such that the following hold for a given string $s \in_{\text{u.a.r.}} \{0,1\}^{C_s \log (\ell / p)}$.
		\begin{enumerate}
			\item  For a string $m$ of length at most $\ell$, we have $\fp{s}{m}{p}{\ell} = (s,f)$, where $f$ is a string of length $C_h \log (1/p)$.
			\item For any bit strings $m$ and $m'$ of length at most $\ell$, if $m = m'$, then $\fp{s}{m}{p}{\ell} = \fp{s}{m'}{p}{\ell}$, else $\Pr \{ \fp{s}{m}{p}{\ell} = \fp{s}{m'}{p}{\ell} \} \leq p$. 
		\end{enumerate}
	\end{mythm}
	
	We refer to $\fp{s}{m}{p}{\ell}$ as the \emph{fingerprint} of the message $m$.
	
	\paragraph{GetPolynomial} Let $\mathcal{M}$ be a multiset of tuples of the form $(x,y) \in GF(q) \times GF(q)$. For each $x \in GF(q)$, we define $\texttt{maj}(\mathcal{M})(x)$ to be the tuple $(x,z)$ that has the highest number of occurrences in $\mathcal{M}$, breaking ties arbitrarily. We define $\texttt{maj}(\mathcal{M}) = \bigcup_{x \in GF(q)}\{ (x,\texttt{maj}(\mathcal{M})(x)) \}$. Given the set $\mathcal{S} = \texttt{maj}(\mathcal{M})$, we define $\getpoly{\mathcal{S}}{d}{q}$ as a function that returns the degree-$d$ polynomial over $GF(q)$ that is supported by the largest number of points in $\mathcal{S}$, breaking ties arbitrarily.\\
	
	\par The following theorem from~\cite{reedSolomon}~\cite{welch1986error} provides conditions under which $\getpoly{\mathcal{S}}{d}{q}$ reconstructs the required polynomial.
	\begin{mythm}
		\label{thm:reedSolomon}
		\cite{reedSolomon}~\cite{welch1986error} Let $P$ be a polynomial of degree $d$ over some field $\mathbb{F}$, and $\mathcal{S} \subset \mathbb{F} \times \mathbb{F}$. Let $g$ be the number of elements $(x,y) \in \mathcal{S}$ such that $y = P(x)$, and let $b = |\mathcal{S}| - g$. Then, if $g > b+d$, we have $\getpoly{\mathcal{S}}{d}{q} = P$. 
	\end{mythm} 
		
	\paragraph{Algebraic Manipulation Detection Codes} Our algorithm also makes use of Algebraic Manipulation Detection (AMD) codes from \cite{amd}. For a given $\eta > 0$, called the strength of AMD encoding, these codes provide three functions: $\texttt{amdEnc}$, $\texttt{amdDec}$ and $\texttt{IsCodeword}$. The function $\amdenc{m}{\eta}$ creates an AMD encoding of a message $m$. The function $\iscodeword{m}{\eta}$ takes a message $m$ and returns true if and only if there exists some message $m'$ such that $\amdenc{m'}{\eta} = m$. The function $\amddec{m}{\eta}$ takes a message $m$ such that $\iscodeword{m}{\eta}$ and returns a message $m'$ such that $\amdenc{m'}{\eta} = m$. These functions enable detection of bit corruption in an encoded message with high probability. The following (slightly reworded) theorem from \cite{amd} helps establish this:
	
	\begin{mythm}
	\label{thm:amd}
		\cite{amd} For any $\eta > 0$, there exist functions $\texttt{amdEnc}$, $\texttt{amdDec}$ and $\texttt{IsCodeword}$, such that for any bit string $m$ of length $x$:
		\begin{enumerate}
			\item $\amdenc{m}{\eta}$ is a string of length $x+C_a \log (1/\eta)$, for some constant $C_a > 0$
			\item $\iscodeword{\amdenc{m}{\eta}}{\eta}$ and $\amddec{\amdenc{m}{\eta}}{\eta}=m$
			\item For any bit string $s\neq 0$ of length $x$, we have \[ \Pr \left( \iscodeword{\amdenc{m}{\eta} \oplus s}{\eta} \right) \leq \eta \]
		\end{enumerate}
	\end{mythm}
	
	With the use of Naor-Naor hash functions along with AMD codes, we are able to provide the required security for messages with Alice and Bob. Assume that the Bob generates the fingerprint $(s,f)$, which upon tampering by the adversary, is converted to $(s \oplus t_1, f \oplus t_2)$ for some strings $t_1,t_2$ of appropriate lengths. Upon receiving this, Alice compares it against the fingerprint of her message $m$ by computing $\fp{s \oplus t_1}{m}{p}{|m|}$, for appropriately chosen $p$. Then, we require that there exist a $\eta \geq 0$ such that for any choice of $t_1,t_2$, 
	\begin{equation*}
		\Pr \{ \fp{s \oplus t_1}{m'}{p}{|m'|} = (s \oplus t_1, f \oplus t_2) \} \leq \eta
	\end{equation*} 
	for any string $m' \neq m$. Theorem~\ref{thm:amd} provides us with this guarantee directly.
				
	\paragraph{Error-correcting Codes} These codes enable us to encode a message so that it can be recovered even if the adversary corrupts a third of the bits. We will denote the encoding and decoding functions by \texttt{ecEnc} and \texttt{ecDec}, respectively. The following theorem, a slight restatement from~\cite{reedSolomon}, gives the properties of these functions.
	
	\begin{mythm}
		\label{thm:ecc}
		\cite{reedSolomon} There is a constant $\eccconst > 0$ such that for any message $m$, we have $|\ecc{m}| \leq \eccconst |m|$. Moreover, if $m'$ differs from $\ecc{m}$ in at most one-third of its bits, then $\ecci{m'}=m$.
	\end{mythm} 
	
	Finally, we observe that the linearity of \texttt{ecEnc} and \texttt{ecDec} ensure that when the error correction is composed with the AMD code, the resulting code has the following properties:
	\begin{enumerate}
		\item If at most a third of the bits of the message are flipped, then the original message can be uniquely reconstructed by rounding to the nearest codeword in the range of \texttt{ecEnc}.
		\item Even if an arbitrary set of bits is flipped, the probability of the change not being recognized is at most $\eta$, \ie\ the same guarantee as the AMD codes.
	\end{enumerate}  
	This is because \texttt{ecDec} is linear, so when noise $\eta$ is added by the adversary to the codeword $x$, effectively what happens is the decoding function $\ecci{x+\eta} = \ecci{x} + \ecci{\eta} = m + \ecci{\eta}$, where $m$ is the AMD-encoded message. But now $\ecci{\eta}$ is an random string that is added to the AMD-encoded codeword.

\paragraph{Silence} In our algorithm, silence on the channel has a very specific meaning. We define the function $\silence{s}$ to return true iff the string $s$ has fewer than $|s|/3$ bit alternations. 

\paragraph{Other notation} We use $\zero{b}$ to denote the $b$-bit string of all zeros, $\odot$ for string concatenation and $\listen{b}$ to denote the function that returns the bits on the channel over the next $b$ time steps. For the sake of convenience, we will use $\log x$ to mean $\lceil \log_2 x \rceil$, unless specified otherwise. Let $\bjconst = \max \{ 19,\naorconst+\amdconst+\eccconst C_s \}$. 

\subsection{Algorithm overview} 
	\begin{algorithm*}[t]
	\caption{Alice's algorithm}
	\label{aliceAlgo}
	\begin{algorithmic}[1]
	\Procedure{Alice}{$M, \delta$}   \Comment{$M$ is a message of length $L$}
		\State $q \gets 2^{\lceil \log L \rceil}$
		\Comment Field size
		\State $d \gets \lceil L / \log q \rceil -1$
		\Comment Degree of polynomial
		\State $P_a \gets \ $degree-$d$ polynomial encoding of $M$ over $GF(q)$
		\State Send $\{ P(0),P(1),\dots,P(d) \}$ 
		\For{$j = 1 \ \text{to }\infty$}
		\Comment{Rounds for the algorithm}
			\State $\eta_j \gets (1/2)^{\lfloor j/d \rfloor}\delta/6d$
			\State $b_j \gets \bjconst \log \left( L/ \eta_j \right)$
			\Comment Message size in this round
			\State $f \gets \ecci{\listen{b_j}}$
			\Comment Fingerprint from Bob
			\If{$\iscodeword{f}{\eta_j}$}
				\State $(s,f_1) \gets \amddec{f}{\eta_j}$
				\If{$(s,f_1) = \fp{s}{P_a}{\eta_j}{L}$}
				\State Send $\ecc{f}$
				\Comment Echo the fingerprint
				\EndIf
			\EndIf
			\State Send $\zero{b_j}$ if the fingerprint was not echoed.
			\State $f_2 \gets \listen{b_j}$
			\If{$\silence{f_2}$}
				\State $\terminate$
				\Comment Bob has likely left
			\Else
				\State $M_a \gets \ $\balls\ of $P_a$ at next two points of the field (cyclically)
				\State Send $\ecc{\amdenc{M_a}{\eta_j}}$
			\EndIf
		\EndFor
   	\EndProcedure
	\end{algorithmic}
	\end{algorithm*}
	
	\begin{algorithm*}[t]
	\caption{Bob's algorithm}
	\label{bobAlgo}
	\begin{algorithmic}[1]
	\Procedure{Bob}{$L, \delta$} 
		\State $q \gets 2^{\lceil \log L \rceil}$
		\Comment Field size
		\State $d \gets \lceil L / \log q \rceil -1$
		\Comment Degree of polynomial
		\State $\mathcal{B} \gets \emptyset$
		\Comment $\mathcal{B} \in GF(q) \times GF(q)$
		\State Listen to first $d+1$ evaluations from Alice 
		\State Add the corresponding \balls\ to $\mathcal{B}$
		\For{$j = 1 \ \text{to }\infty$}
			\State $\eta_j \gets (1/2)^{\lfloor j/d \rfloor}\delta/6d$
			\State $b_j \gets \bjconst \log \left( L/ \eta_j \right)$
			\Comment Message size in this round
			\State $P_b \gets \getpoly{\texttt{maj}(\mathcal{B})}{d}{q}$
			\State Sample a string $s \in_{\text{u.a.r.}} \{ 0,1\}^{C_s b_j/C}$
			\State $f_b \gets \amdenc{\fp{s}{P_b}{\eta_j}{L}}{\eta_j}$
			\State Send $\ecc{f_b}$
			\Comment Send Alice the fingerprint of the polynomial
			\State $f_b' = \ecci{\listen{b_j}}$
			\Comment Listen to Alice's echo
			\If{$f_b' = f_b$}
				\State $\terminate$
			\Else
				\State Send a string $f_2' \in_{\text{u.a.r.}} \{0,1\}^{b_j}$
				\State Receive \balls\ for the next two \bins\ and add to $\mathcal{B}$ 
			\EndIf
		\EndFor
	\EndProcedure
	\end{algorithmic}
	\end{algorithm*}
	
	Our algorithm for the case when $L$ is known is given in two parts: Algorithm~\ref{aliceAlgo} is what Alice follows and Algorithm~\ref{bobAlgo} is what Bob follows. Both algorithms assume knowledge of the message length $L$ and the error tolerance $\delta$. The idea is for Alice to compute a degree-$d$ polynomial encoding of $M$ over a field of size $q$. Here $q = 2^{\lceil \log L \rceil}$ and $d = \lceil L / \log q \rceil -1$. She begins by sending evaluations of this polynomial over the first $d+1$ field elements to Bob in plaintext, which Bob uses to reconstruct the polynomial and retrieve the message. He also computes a fingerprint of this polynomial and sends it back to Alice. He encodes this fingerprint with AMD encoding and then ECC encoding, so that any successful tampering will require at least a third of the bits in the encoded fingerprint to be flipped and will be detected with high probability. If Alice receives a correct fingerprint, she echoes it back to Bob. Upon listening to this echo, Bob terminates. The channel from Bob to Alice is now silent, after incepting which Alice terminates the protocol as well. 
	
	If the adversary flips bits on the channel so that Bob's fingerprint mismatches, Alice recognizes this mismatch with high probability and exchanges more evaluations of her polynomial with Bob, proceeding in rounds. In each round, Alice sends two more evaluations of the polynomial on the next two field elements and sends them to Bob. Bob uses these to reconstruct his polynomial and sends a fingerprint back to Alice. The next round only begins if Alice did not terminate in this round, which will require this fingerprint to match and for Alice to intercept silence after Bob has terminated. We will bound the number of rounds and the failure probability for our algorithm in the next section.
	
	\subsection{Example Run}
	
	\begin{figure*}
		\includegraphics[width=500pt]{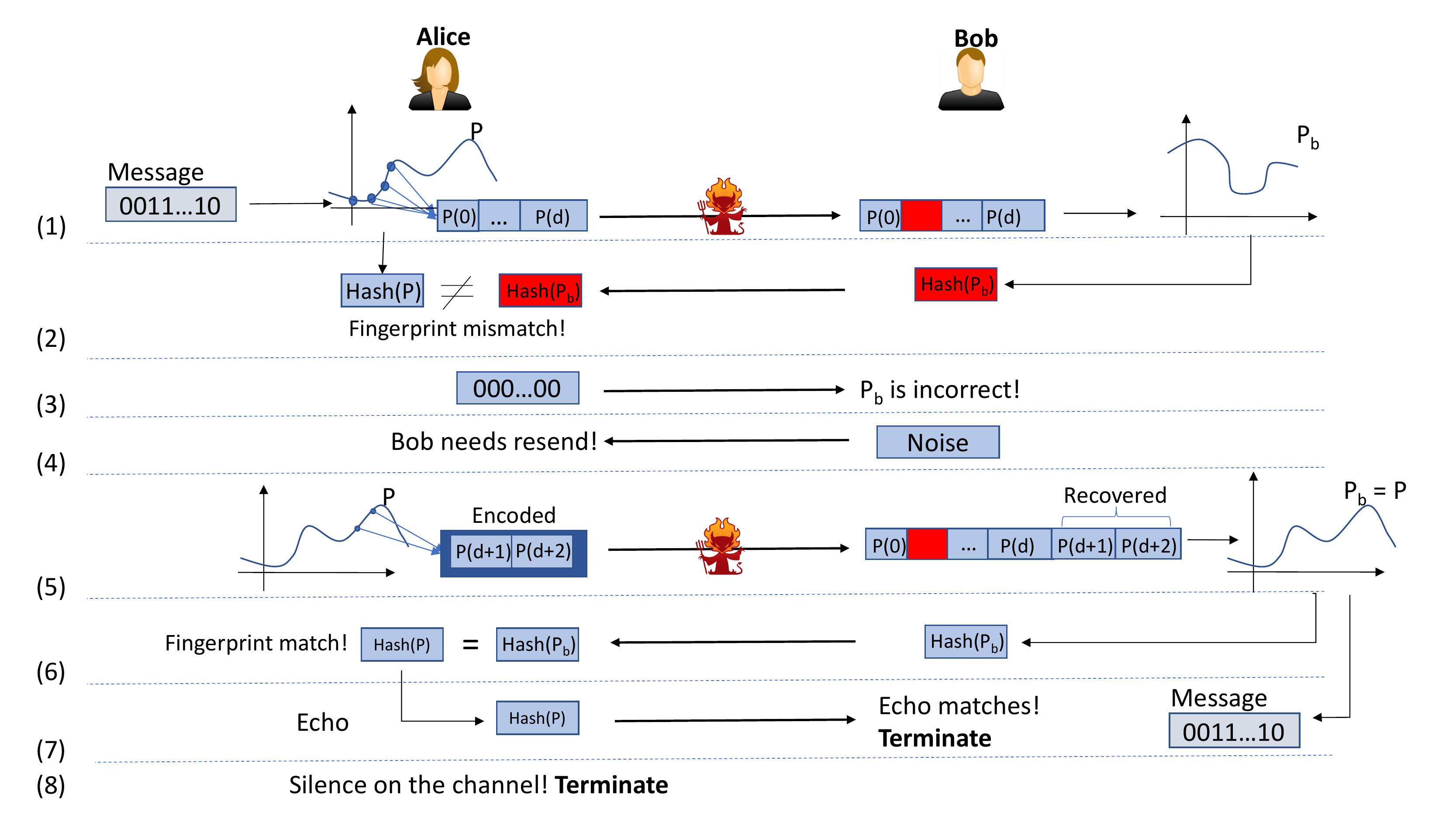}
		\caption{Example run of our protocol for the case when the adversary corrupts one polynomial evaluation tuple in plaintext and fewer than a third of the bits in the encoded tuples that are sent during the resend. The blue boxes represent bits from our protocol, red boxes represent bits flipped by the adversary, and the dar blue box emphasizes the fact that the contained bits are encoded using ECC and AMD codes.}
		\label{fig:exampleRun}
	\end{figure*}
	
	We now discuss an example of a run of our protocol to make the different steps in the algorithm more clear. We illustrate this example in Fig.~\ref{fig:exampleRun} and provide a step-by-step explanation below. 
	\begin{enumerate}
	\item Alice begins by computing a polynomial $P$ corresponding to the message and sends its evaluation on the first $d+1$ field elements to Bob, in plaintext. The adversary now corrupts one of the evaluation tuples so that the polynomial $P_b$ that Bob reconstructs is different than $P$.
	\item Bob computes the fingerprint of this polynomial, depicted $Hash(P_b)$ for brevity, and sends it to Alice. Alice compares this fingerprint against the hash of her own polynomial, $Hash(P)$, and notices a mismatch. 
	\item In response, Alice remains silent. Bob is now convinced that his version of the polynomial is incorrect, so he sends noise to Alice to ask her for a resend. 
	\item Alice encodes two more evaluations of $P$ at the next two field elements and sends them to Bob. The adversary tries to tamper with these evaluations by flipping some bits. For this example, we assume that he flips fewer than a third of the total number of bits in the encoded evaluations. Upon decoding, Bob is able to successfully recover both the evaluations and uses the $\texttt{GetPolynomial}$ subroutine to recompute $P_b$, which in this case matches $P$. 
	\item Bob computes $Hash(P_b)$ and sends it to Alice. Upon seeing this hash and verifying that it matches $Hash(P)$, Alice is now convinced that Bob has the correct copy of the polynomial, and hence, the original message.
	\item Alice echoes the hash back to Bob, upon hearing which Bob extracts the message from the polynomial (using its coefficients) and terminates the protocol. Silence follows on the channel from Bob to Alice.
	\item Alice intercepts silence and terminates the protocol as well.
	\end{enumerate} 
	The message has now successfully been transmitted from Alice to Bob.
	
	\section{Analysis}
	\label{sec:analysisKnownL}
	We now prove that our algorithm is correct with probability at least $1-\delta$, and compute the number of bits sent. Before proceeding to the proof, we define three bad events:
	\begin{enumerate}
		\item \emph{Unintentional Silence.} When Bob executes step 18 of his algorithm, the string received by Alice is interpreted as silence.
		\item \emph{Fingerprint Error.} Fingerprint hash collision as per Theorem~\ref{thm:hash}. 
		\item \emph{AMD Error.} The adversary corrupts an AMD encoded message into an encoding of a different message. 
	\end{enumerate}
	
	\paragraph{Rounds}
	For both Alice and Bob, we define a \emph{round} as one iteration of the \texttt{for} loop in our algorithm. We refer to the part of the algorithm before the \texttt{for} loop begins as \emph{round $0$}. The AMD encoding strength $\eta$ is equal to $\delta/6d$ initially and decreases by a factor of $2$ every $d$ rounds. This way, the number of bits added to the messages increases linearly every $d$ rounds, which enhances security against corruption. 
	
	\subsection{Correctness and Termination}
	We now prove that with probability at least $1-\delta$, Bob terminates the algorithm with the correct guess of Alice's message.
	
	\subsubsection{Unintentional Silence} The following lemmas show that Alice terminates before Bob with probability at most $\delta/3$.	
	
	\begin{lemma}
		\label{lem:chernoff}
		For $b\geq 71$, the probability that a $b$-bit string sampled uniformly at random from $\{0,1 \}^b$ has fewer than $b/3$ bit alternations is at most $e^{-b/19}$.
	\end{lemma}
	\begin{proof}
		Let $s$ be a string sampled uniformly at random from $\{ 0,1\}^b$, where $b \geq 71$. Denote by $s[i]$ the $i^{th}$ bit of $s$. Let $X_i$ be the indicator random variable for the event that $s[i]\neq s[i+1]$, for $1 \leq i < b$. Note that all $X_i$'s are mutually independent. Let $X$ be the number of bit alternations in $s$. Clearly, $X = \sum_{i=1}^{b-1}X_i$, which gives $\mathbb{E}(X) = \sum_{i=1}^{b-1}\mathbb{E}(X_i)$, using the linearity of expectation. Since $\mathbb{E}(X_i)=1/2$ for all $1 \leq i < b$, we get $\mathbb{E}(X) = (b-1)/2$. Using the multiplicative version of Chernoff bounds~\cite{dubhashi2009concentration} for $0 \leq t \leq \sqrt{b-1}$,
		\begin{equation*}
			\Pr \left \{ X < \frac{b-1}{2} - \frac{t\sqrt{b-1}}{2} \right \} \leq e^{-t^2 /2}.
		\end{equation*}
		To obtain $\Pr \{ X < b/3 \}$, set $t = \frac{b-3}{3\sqrt{b-1}}$ to get,
		\begin{equation*}
		\Pr \{ X < b/3 \} \leq e^{-\frac{(b-3)^2}{18(b-1)}} \leq e^{- b/19} \quad \text{for $b \geq 71$}.
		\end{equation*}
	\end{proof} 

	\begin{lemma}
	\label{lem:term}
   	Alice terminates the algorithm before Bob with probability at most $\delta/3$.
	\end{lemma}
	\begin{proof}
	Let $\xi$ be the event that Alice terminates before Bob. This happens when the string sent by Bob in step 18 after possible adversarial corruptions is interpreted as silence by Alice. Let $\xi_j$ be the event that Alice terminates before Bob in round $j$ of the algorithm. Then, using a union bound over the rounds, the fact that $C \geq 19$ and Lemma~\ref{lem:chernoff}, we get 
	\begin{equation*}
	\begin{split}
		\Pr \{\xi \} &\leq \sum_{j\geq 1} \Pr \{ \xi_j \} \leq \sum_{j \geq 1}e^{-b_j/19} \leq \sum_{j \geq 1}2^{-b_j/19} \\
		&= \sum_{j \geq 1}2^{-C\log (L/\eta_j)/19} \leq \sum_{j \geq 1}2^{-\log (L/\eta_j)} = \sum_{j \geq 1}\log (\eta_j /L) \\
		&\leq \frac{\delta}{6Ld} \sum_{j \geq 0} \left( \frac{1}{2} \right)^{\lfloor j/d \rfloor} \leq \frac{\delta}{3L} \leq \frac{\delta}{3}
	\end{split}
	\end{equation*}
	Note that Lemma~\ref{lem:chernoff} is applicable here because for each $j \geq 1$, we have $b_j \geq 71$. To see this, use the fact that $d \leq 2L/\log L$ and $\delta < 1$ to obtain the condition $L^2 \geq 2^{71/C}/12$, which is always true because $L^2 > 4 > 2^{71/C}/12$.
	\end{proof}
	
	\subsubsection{Fingerprint Failure} The following lemma proves that the fingerprint error happens with probability at most $\delta/3$, ensuring the correctness of the algorithm.
	
	\begin{lemma}
	\label{lem:correctness}
	Upon termination, Bob does not have the correct guess of Alice's message with probability at most $\delta/3$.	
	\end{lemma}
	\begin{proof}
		Let $\xi$ be the event that Bob does not have the correct guess of Alice's message upon termination. Note that in round $j$, from Theorem~\ref{thm:hash}, the fingerprints fail with probability at most $\eta_j$. Using a union bound over these rounds, we get
		\begin{equation*}
			\begin{split}
				\Pr \{ \xi \} \leq \sum_{j \geq 1}\eta_j = \sum_{j \geq 1} \frac{\delta}{6d} \left( \frac{1}{2} \right)^{\left \lfloor j/d \right \rfloor} \leq \frac{\delta}{6} \sum_{j \geq 0} (1/2)^{j} = \frac{\delta}{3}
			\end{split}
		\end{equation*}
	\end{proof}
	
	\subsubsection{AMD Failure} 
	\begin{lemma}
		\label{lem:amdFail}
		The probability of AMD failure is at most $\delta/3$.
	\end{lemma}
	\begin{proof}
		Note that in round $j$, from Theorem~\ref{thm:amd}, AMD failure occurs with probability at most $\eta_j$. Hence, using a union bound over the rounds, the AMD failure occurs with probability $\sum_{j \geq 1}\eta_j = \sum_{j \geq 1} \frac{\delta}{6d} \left( \frac{1}{2} \right)^{\left \lfloor j/d \right \rfloor} \leq \frac{\delta}{6} \sum_{j \geq 0} (1/2)^{j} = \frac{\delta}{3}$.
	\end{proof}

	\subsection{Probability of Failure}

	\begin{lemma}
	  \label{lem:prob}
	Our algorithm succeeds with probability at least $1 - \delta$. 
		\end{lemma}
	\begin{proof}
	Lemmas~\ref{lem:term},~\ref{lem:correctness} and~\ref{lem:amdFail} ensure that the three bad events, as defined previously, each happen with probability at most $\delta/3$. Hence, using a union bound over the occurrence of these three events, the total probability of failure of the algorithm is at most $\delta$. If the three bad events do not occur, then Alice will continue to send evaluations of the polynomial until Bob has the correct message. Since $T$ is finite, Bob will eventually have the correct message and terminate.
	\end{proof}

	\subsection{Cost to the algorithm}
	Recall that Alice and Bob compute their polynomials $P_a$ and $P_b$, respectively, over $GF(q)$. We refer to every $(x,y) \in GF(q) \times GF(q)$ that Bob stores after receiving the  evaluation $y$, that has potentially been tampered with, of the polynomial $P_a$ at $x$ from Alice as a \emph{\ball}. We call a \ball\ $(x,y)$ in Bob's set $\mathcal{B}$ \emph{good} if $P_a(x)=y$ and \emph{bad} otherwise.\\ 
	
	\par We begin by stating two important lemmas that relate the number of bits flipped by the adversary to make $m$ \balls\ bad to the number of bits required to send them.
	
	\begin{lemma}
	\label{lem:corr}
		Let $f(m)$ be the number of bits flipped by the adversary to make $m$ \balls\ bad. Then, $f(m) \geq m$ if $m \leq d+1$, and \[ f(m) \geq (d+1)+\frac{\bjconst}{6}\left( (m-d-1)\log (6Ld/\delta) + \frac{(m-d-3)^2}{4d}\right) \]otherwise.
	\end{lemma}
	\begin{proof}
	Let $m = m_1 + m_2$, where $m_1 \leq d+1$ is the number of \balls\ that were not encoded and $m_2$ is the number of AMD and error-encoded \balls. Clearly, $f(m_1) = m_1$. Each of the remaining $m_2$ \balls\ are sent in pairs, one pair per round. Since the adversary needs to flip at least a third of the number of bits for each encoded \ball\ to make it bad, we have
	\begin{equation*}
		\begin{split}
			f(m) &\geq m_1 + \frac{1}{3}\sum_{j = 1}^{m_2/2}b_j \\
			&= m_1 + \frac{C}{3}\sum_{j = 1}^{m_2/2} \left( \log \left( \frac{6Ld}{\delta} \right) + \left \lfloor \frac{j}{d} \right \rfloor \right) \\
			&\geq m_1 + \frac{C}{6}\left( m_2 \log \left( \frac{6Ld}{\delta} \right) + \frac{(m_2-2)^2}{4d} \right)
		\end{split}
	\end{equation*}
	Since the number of bits per \ball\ increases monotonically, the expression above becomes $f(m) \geq m$ if $m \leq d+1$, and \[ f(m) \geq (d+1)+\frac{\bjconst}{6}\left( (m-d-1)\log (6Ld/\delta) + \frac{(m-d-3)^2}{4d}\right) \]otherwise.
	\end{proof}

	\begin{lemma}
	\label{lem:cost}
		Let $g(m)$ be the number of bits required to send $m$ \balls, where $m \geq d+1$. Then,
	\begin{equation*}
	g(m) \leq L + 5C\left( \frac{(m-d-1)}{2}\log(6Ld/\delta) + \frac{(m-d+1)^2}{8d} \right).
	\end{equation*}
	\end{lemma}
	\begin{proof}
	If $m < d+1$, then we have $g(m) = m\log q \leq L$, since each of these $m$ \balls\ is of length $\log q$. For $m>d+1$, taking into account the fact that each round involves exchange of at most $5$ messages between Alice and Bob, we get
	\begin{equation*}
		\begin{split}
			g(m) &\leq L + 5\sum_{j=1}^{(m-d-1)/2}b_j \\
			&= L + 5C\sum_{j=1}^{(m-d-1)/2} \left( \log \left( \frac{6Ld}{\delta} \right) + \left \lfloor \frac{j}{d} \right \rfloor \right) \\
			&\leq L + 5C\left( \frac{(m-d-1)}{2}\log(6Ld/\delta) + \frac{(m-d+1)^2}{8d} \right)
		\end{split}
	\end{equation*}
	\end{proof}
	
	\begin{lemma}
		\label{lem:rounds}
		Let $L \geq 3$, and $r$ be any round at the end of which $P_b \neq P_a$. Then the number of bad \balls\ through round $r$ is at least $r/4$. 
		\end{lemma}
		\begin{proof}
		We call a \bin\ $x \in GF(q)$ \emph{good} if $(x,P_a(x)) \in \texttt{maj}(\mathcal{B})$, and \emph{bad} otherwise. Let $g_e$ be the number of good \bins\ and $b_e$ be the number of bad \bins\ up to round $r$. Similarly, let $g_t$ be the number of good \balls\ and $b_t$ be the number of bad \balls\ up to round $r$. Then, from Theorem~\ref{thm:reedSolomon}, we must have $b_e \geq g_e - d$. Note that the total number of \bins\ for which Bob has received \balls\ from Alice through round $r$ is $b_e + g_e = \min (d+2r+1,q)$. Adding this equality to the previous inequality, we have 
		\begin{equation}
		\label{eq:eq2}
		b_e \geq \frac{1}{2}\min (2r+1,q-d).	
		\end{equation}

		The total number of \balls\ received by Bob up to round $r$ is given by
		\begin{equation}
		\label{eq:eq4}
			b_t + g_t = d+2r+1.
		\end{equation}
		Note that every bad \bin\ is associated with at least $\left \lfloor \frac{b_t + g_t}{2 (b_e + g_e)} \right \rfloor$ \balls. This gives $b_t \geq b_e\left \lfloor \frac{b_t + g_t}{2 (b_e + g_e)} \right \rfloor$. Using this inequality with Eqs.~\eqref{eq:eq2} and~\eqref{eq:eq4}, we have
		\begin{equation}
		\begin{split}
			b_t &\geq \frac{1}{2}\min (2r+1,q-d) \left \lfloor \frac{d+2r+1}{2 \min (d+2r+1,q)} \right \rfloor \\
			&\geq \frac{1}{2} \left \lfloor \frac{d+2r+1}{2 \min (d+2r+1,q)} \min (2r+1,q-d) \right \rfloor
		\end{split}
		\end{equation}
		\textbf{Case I: $\mathbf{(d+2r+1 \leq q)}$} For this case, we have 
		\begin{equation}
			\label{eq:4a}
			\frac{1}{2}\left \lfloor \frac{d+2r+1}{2 \min (d+2r+1,q)}\min (2r+1,q-d) \right \rfloor = \frac{1}{2}\left \lfloor \frac{2r+1}{2} \right \rfloor \geq \frac{r}{4}
		\end{equation}
		\textbf{Case II: $\mathbf{(d+2r+1 > q)}$} For this case, we have 
		\begin{equation}
			\label{eq:4b}
			\begin{split}
			\frac{1}{2}\left \lfloor \frac{d+2r+1}{2 \min (d+2r+1,q)}\min (2r+1,q-d) \right \rfloor &= \frac{1}{2}\left \lfloor \frac{(d+2r+1)(q-d)}{2q} \right \rfloor  \\
			&\geq \frac{1}{2}\left \lfloor \frac{2r+1}{2}\left( 1 - \frac{d}{q} \right) \right \rfloor \\
			&\geq \frac{r}{4}
			\end{split}
		\end{equation}	
		where the last inequality holds since $d/q \leq 1/3$ for $L \geq 3$.\\
		\par Combining Eqs.~\eqref{eq:4a} and~\eqref{eq:4b}, we get $b_t \geq r/4$.
		\end{proof}
	
	\par We now state a lemma that is crucial to the proof of Theorem~\ref{thm:mainThm}.
	
	\begin{lemma}
	  \label{lem:totalcost}
	If Bob terminates before Alice, the total number of bits sent by our algorithm is \[ L + O \left( T + \min \left(T+1,\frac{L}{\log L} \right) \log \left( \frac{L}{\delta} \right) \right). \]
	\end{lemma}
	\begin{proof}
		Let $r'$ be the last round at the end of which $P_b \neq P_a$, or $0$ if $P_b = P_a$ at the end of round $1$ and for all subsequent rounds. Let $T_1$ be the number of bits corrupted by the adversary through round $r'$. Let $A_1$ represent the total cost through round $r'$ and $A_2$ be the cost of the algorithm after round $r'$. Note that after round $r'$, the adversary must corrupt one of either (1) the fingerprint, or (2) its echo, or (3) silence on the channel in Step 15 of Alice's algorithm, in every round to delay termination. Also, after round $r'$, Alice and Bob must exchange at least a fingerprint and an echo even if $T=0$. Thus, we have,
		\begin{equation}
		\label{eq:a2}	
		A_2 = O(T + \log(L/\delta))
		\end{equation}
	
		\par Recall that the number of \balls\ sent up to round $r'$ is $d+2r'+1$. Then, from Lemma~\ref{lem:cost}, we have 
		\begin{equation}
		\label{eq:eq10}
		\begin{split}
		A_1 &\leq g(d+2r'+1) \\
		&\leq L + 5C\left( r'\log(6Ld/\delta) + \frac{(r'+1)^2}{2d} \right).
		\end{split}
		\end{equation} 
		
		 From Lemma~\ref{lem:rounds}, we have that the number of bad \balls\ is at least $\lceil r'/4 \rceil$. Thus, from Lemma~\ref{lem:corr}, we have $T_1 \geq f(\lceil r'/4 \rceil)$, which implies $T_1 \geq r'/4$ if $r'/4 \leq d+1$. Otherwise, we have
		\begin{equation}
		\label{eq:eq6}
		T_1 \geq 
	 (d+1)+\frac{C}{6}\left( (r'/4 - d - 1)\log (6Ld/\delta) + \frac{(r'/4-d+3)^2}{4d} \right)
		\end{equation}
		\\
		\textbf{Case I : $\mathbf{(r'/4 \leq d+1)}$} Since $T_1$ is at least the number of bad \balls, from Lemma~\ref{lem:rounds}, we have $T_1 \geq r'/4$, which gives $r' \leq \min (4T_1, 4(d+1))$. Hence, using Eq~\eqref{eq:eq10}, we get,
		\begin{align}
			A_1 &\leq L + 5C\left( r'\log(6Ld/\delta) + \frac{(r'+1)^2}{2d} \right) \nonumber \\
			&\leq L + 5C\left( \min (4T_1, 4(d+1)) \log(6Ld/\delta) + \frac{(4d+5)^2}{2d} \right) \nonumber \\
			&= L + O \left( \min \left(T_1, \frac{L}{\log L} \right)\log (L/\delta) + \frac{L}{\log L}\right) \label{eq:eq11}
		\end{align}
		where the last equality holds because $d \leq L/\log L+1$. \\
		\\
		\textbf{Case II : $\mathbf{(r'/4 > d+1)}$} From Eq.~\eqref{eq:eq6}, we have 
		\begin{equation}
		\label{ineq:t1}
			T_1 \geq (d+1)+\frac{C}{6}\left( (r'/4 - d - 1)\log (6Ld/\delta) + \frac{(r'/4-d+3)^2}{4d} \right).
		\end{equation}
		Since each summand in the inequality above is positive and $C > 6$, we get $(r'/4 - d - 1)\log (6Ld/\delta) \leq T_1$, which gives 
		\begin{equation}
			\label{eq:eq12a}
			r'\log(6Ld/\delta) \leq 4T_1 + 4(d+1)\log(6Ld/\delta).
		\end{equation}
		Since $\frac{(r'/4-d+3)^2}{4d} \leq T_1$, we have $r' \leq 8\sqrt{T_1d} + 4d - 12$. Building on this, we get,
		\begin{equation}
			\label{eq:eq12b}
			\frac{(r'+1)^2}{2d} \leq \frac{\left( 8\sqrt{T_1d} + 4d - 11 \right)^2}{2d}
		\end{equation}
		Hence, from Eqs.~\eqref{eq:eq10},~\eqref{eq:eq12a} and~\eqref{eq:eq12b} , we get 
	\begin{align}
				A_1 &\leq L + 5C\left( r'\log(6Ld/\delta) + \frac{(r'+1)^2}{2d} \right) \nonumber \\
				&\leq L + 5C \left( 4T_1 + 4(d+1)\log(6Ld/\delta) + \frac{\left( 8\sqrt{T_1d} + 4d - 11 \right)^2}{2d} \right) \nonumber\\
				&= L + O \left( T_1 + \left( \frac{L}{\log L} \right)\log(L/\delta) \right) \label{eq:eq13}
			\end{align}
	where the last equality holds because $d \leq L/\log L+1$ and $T_1 \geq d+1$ from inequality~\eqref{ineq:t1}.\\
	
	\par Combining Eqs.~\eqref{eq:a2},~\eqref{eq:eq11} and~\eqref{eq:eq13}, the total number of bits sent by the algorithm becomes \[ A_1 + A_2 = L + O \left( T + \min \left(T+1,\frac{L}{\log L} \right) \log \left( \frac{L}{\delta} \right) \right)\]
	\end{proof}
	Putting it all together, we are now ready to state our main theorem.
	
	\begin{mythm}
		\label{thm:mainThm2}
		Our algorithm tolerates an unknown number of adversarial errors, $T$, and for a given $\delta \in (0,1)$, succeeds with probability at least $1-\delta$, and sends $L + O \left( T + \min \left(T+1,\frac{L}{\log L} \right) \log \left( \frac{L}{\delta} \right) \right)$ bits.	
			
	\end{mythm}	\begin{proof}
		By Lemmas~\ref{lem:prob}, with probability at least $1-\delta$, Bob terminates before Alice with the correct message. If this happens, then by Lemma~\ref{lem:totalcost}, the total number of bits sent is \[ L + O \left( T + \min \left(T+1,\frac{L}{\log L} \right) \log \left( \frac{L}{\delta} \right) \right) \]
	\end{proof}

\section{Unknown $L$}
\label{sec:unknownL}
We now discuss an algorithm for the case when the message length $L$ is unknown to Bob. The only parameter now known to both Alice and Bob is $\delta$.

Our main idea is to make use of an algorithm from~\cite{aggarwal2016secure}, which enables Alice to send a message of unknown length to Bob in our model, but is inefficient. \footnote{We refer the reader to~\cite{aggarwal2016secure} for details on this algorithm; we discuss only its use in this paper.}  We thus use a two phase approach.  First, we send the \emph{length} of the message $M$ (i.e. a total of $\log L$ bits) from Alice to Bob using the algorithms of~\cite{aggarwal2016secure}.  Second, once Bob learns the value $L$, we use the algorithm from Section~\ref{sec:knownL} to communicate the message $M$.  We will show that the total number of bits sent by this two phase algorithm is asymptotically similar to the case when the message length is known by Bob in advance.

\subsection{Algorithm Overview}
Let $\pi_1$ be a noise-free protocol in which Alice sends $L$ to Bob, who is unaware of the length ($\log L$ in this case) of the message. Let $\pi_2$ be a noise-free protocol in which Alice sends $M$ to Bob, who knows the length $L = |M|$ \textit{a priori}. W can write the noise-free protocol $\pi$ to communicate $M$ from Alice to Bob, who does not know $L$, as a composition of $\pi_1$ and $\pi_2$ in this order. Let $\pi'_1, \pi'_2$ and $\pi'$ be the simulations of $\pi_1,\pi_2$ and $\pi$, respectively, that are robust to adversarial bit flipping.

To simulate $\pi'$ with desired error probability $\delta > 0$, we proceed in two steps. We first make $\pi_1$ robust with $\delta_1 = \delta/2$ error tolerance using Algorithm $3$ from ~\cite{aggarwal2016secure}, setting $n = 2$. Then, we make $\pi_2$ robust with $\delta_2 = \delta/2$ error tolerance using Algorithms~\ref{aliceAlgo} and~\ref{bobAlgo}. This way, when we compose the robust versions of $\pi_1$ and $\pi_2$, we get $\pi'$ with error probability at most $\delta_1 + \delta_2 = \delta$ (by union bound). The correctness of $\pi'$ immediately follows from the correctness of $\pi'_1$ and $\pi'_2$, by construction.

\subsection{Probability of Failure}
The failure events for $\pi'$ are exactly the failure events for $\pi'_1$ and $\pi'_2$. In other words, we say $\pi'$ fails when one or both of $\pi'_1$ and $\pi'_2$ fail. Thus, the failure probability of $\pi'$ is at most $\delta/2 + \delta/2 = \delta$, by a simple union bound over the two sub-protocols.

\subsection{Number of bits sent}
To analyze the number of bits sent, let $T_1$ be the number of bits flipped by the adversary in $\pi'_1$ and $T_2$ be the number of bits flipped by the adversary in $\pi'_2$. Recall that the length of the message from Alice to Bob in $\pi'_1$ is $\log L$ and that in $\pi'_2$ is $L$. Let $A_1$ be the number of bits sent in $\pi'_1$ and $A_2$ be the number of bits sent in $\pi'_2$. Thus, using Theorem $1.1(2)$ from~\cite{aggarwal2016secure} (with $n=2, L = \log L, T = T_1, \delta_1 = \delta/2$ and $\alpha = 1$), we get \[ A_1 = O \left( \log L \cdot \log \log L + T_1 \right)\] Similarly, using Theorem~\ref{thm:mainThm2} from this paper (with $\delta_2 = \delta/2 $), we get \[A_2 = L + O\left( T_2 + \min \left( T_2+1, L/\log L \right) \log L) \right)\] Using $T = T_1 + T_2$, the total number of bits sent by $\pi'$ is then $A_1 + A_2 = L + O\left( T + \min \left( T+1, L/\log L \right) \log L \right)$.  The proof of Theorem~\ref{thm:mainThm} now follows directly from the above analysis.

Note that another approach to sending a message of unknown length from Alice to Bob would have been to directly use the algorithm in~\cite{aggarwal2016secure} with $n=2$. However, this would have incurred a higher blowup than the approach that we take in this paper. More specifically, when $T$ is small, the direct use of the multiparty algorithm gives a multiplicative logarithmic blowup in the number of bits, while our current approach maintains the constant overall blowup in the number of bits by using the heavy weight protocol for the length of the message instead (which is exponentially smaller than the message). 

\section{Conclusion}
\label{sec:conclusion}
	We have described an algorithm for sending a message over a two-way noisy channel.  Our algorithm is robust to an adversary that can flip an unknown but finite number of bits on the channel.  The adversary knows our algorithm and the message to be sent, but does not know the random bits of the sender and receiver, nor the bits sent over the channel.  The receiver of the message does not know the message length in advance.  
	
	Assume the message length is $L$, the number of bits flipped by the adversary is $T$, and $\delta > 0$ is an error parameter known to both players.  Then our algorithm sends an expected number of bits that is $L + O \left( T + \min \left(T+1,\frac{L}{\log L} \right) \log \left( \frac{L}{\delta} \right) \right)$, and succeeds with probability at least $1 - \delta$.  When $T = \Omega(L)$ and $\delta$ is polynomially small in $L$, the number of bits sent is $L + O\left( T \right)$, which is asymptotically optimal; and when $T = o(L/\log L)$, the number of bits sent is $L+o(L)$. 

Many open problems remain including the following.  First, Can we determine asymptotically matching upper and lower bounds on the number of bits required for our problem?  Our current algorithm is optimal for $T = \Omega(L)$, and seems close to optimal for $T = O(1)$, but is it optimal for intermediate values of $T$?  Second, Can we tolerate a more powerful adversary or different types of adversaries?  For example, it seems like our current algorithm can tolerate a completely omniscient adversary, if that adversary can only flip a chosen bit with some probability that is $1-\epsilon$ for some fixed $\epsilon>0$.  Finally, can we extend our result to the problem of sending our message from a source to a target in an arbitrary network where nodes are connected via noisy two-way channels? This final problem seems closely related to the problem of network coding~\cite{liew2013physical,matsuda2011survey,bassoli2013network}, for the case where the amount of noise and the message size is not known in advance.  In this final problem, since there are multiple nodes, we would likely also need to address problems of asynchronous communication. 
	
\bibliographystyle{ACM-Reference-Format}
\bibliography{ref.bib} 

\end{document}